\documentclass[letterpaper]{article} 
\usepackage{aaai2026}  
\usepackage{times}  
\usepackage{helvet}  
\usepackage{courier}  
\usepackage[hyphens]{url}  
\usepackage{graphicx} 
\usepackage{natbib}  
\usepackage{caption} 
\usepackage{amsthm}
\usepackage{booktabs}    
\usepackage{graphicx}    
\usepackage{array}       
\usepackage{multirow}    
\usepackage{diagbox}

\usepackage{amsmath}     

\usepackage{enumitem}
\usepackage{thmtools}
\usepackage[ruled]{algorithm2e}
\usepackage{float}  

\usepackage{cleveref}
\usepackage{newfloat}
\usepackage{listings}
\DeclareCaptionStyle{ruled}{labelfont=normalfont,labelsep=colon,strut=off} 
\floatstyle{ruled}
\newfloat{listing}{tb}{lst}{}
\floatname{listing}{Listing}
\declaretheorem{lemma}

\declaretheorem{fact}

\theoremstyle{remark}
\newtheorem{remark}{\textbf{Remark}}

\theoremstyle{definition}

\newtheorem{condition}{Condition}

\usepackage [operators,sets,landau,complexity]{cryptocode}

\usepackage{pifont}
\usepackage{bm}
\renewcommand{\vec}[1]{\bm{#1}}
\DeclareMathOperator*{\E}{\mathbb{E}}

\def\NN{{\mathbb N}}
\def\RR{{\mathbb R}}

\def\mM{{\mathcal M}}

\def\mA{{\mathcal A}}

\usepackage{amssymb}  
\usepackage{physics2}
\usephysicsmodule{ab}
\usepackage{derivative}
\usepackage{blkarray}
\newcommand{\inprod}[2]{\left\langle #1, #2 \right\rangle}

\usepackage{xspace}
\newcommand{\maco}
{\texttt{MACO}\xspace}

\title{A Multi-Agent Conversational Bandit Approach to Online Evaluation and Selection of User-Aligned LLM Responses}

\author{
Xiangxiang Dai\textsuperscript{\rm 1},
Yuejin Xie\textsuperscript{\rm 2},
Maoli Liu\textsuperscript{\rm 1},
Xuchuang Wang\textsuperscript{\rm 3},\\
Zhuohua Li\textsuperscript{\rm 4, 1}\thanks{Zhuohua Li is the corresponding author.},
Huanyu Wang\textsuperscript{\rm 5},
John C.S. Lui\textsuperscript{\rm 1}
}
\affiliations{
\textsuperscript{\rm 1}The Chinese University of Hong Kong\\
\textsuperscript{\rm 2}Huazhong University of Science and Technology\\
\textsuperscript{\rm 3}University of Massachusetts Amherst\\
\textsuperscript{\rm 4}Guangzhou Institute of Technology, Xidian University\\
\textsuperscript{\rm 5}Huawei Technologies Co., Ltd.\\[0.3em]
Email: \{xxdai23, mlliu, cslui\}@cse.cuhk.edu.hk, yuejinxie@hust.edu.cn,\\
xuchuangwang@umass.edu, lizhuohua@xidian.edu.cn, huanyuhello@zju.edu.cn
}

\usepackage{bibentry}

\begin{document}

\maketitle

\begin{abstract}
Prompt-based offline methods are commonly used to optimize large language model (LLM) responses, but evaluating these responses is computationally intensive and often fails to accommodate diverse response styles. This study introduces a novel online evaluation framework that employs a multi-agent conversational bandit model to select optimal responses while aligning with user preferences dynamically.
To tackle challenges such as high-dimensional features, large response sets, adaptive conversational needs, and multi-device access, we propose \maco, Multi-Agent Conversational Online Learning, which comprises two key components: (1) \texttt{MACO-A}: Executed by local agents, it employs an online elimination mechanism to filter out low-quality responses.
(2) \texttt{MACO-S}: Executed by the cloud server, it adaptively adjusts selection strategies based on aggregated preference data. An adaptive preference mechanism triggers asynchronous conversations to enhance alignment efficiency. Theoretical analysis demonstrates that \maco achieves near-optimal regret bounds, matching state-of-the-art performance in various degenerate cases. 
Extensive experiments utilizing Google and OpenAI text embedding models on the real-world datasets with different response styles, combined with Llama and GPT-4o, show that \maco consistently outperforms baseline methods by at least 8.29\% across varying response set sizes and numbers of agents.\footnote{Implementation Code: https://github.com/TarferSoul/MACO} 

\end{abstract}

\section{Introduction}
\label{sec:introduction}

Large Language Models (LLMs) have profoundly transformed the technological landscape of society~\cite{zhusoft2025,han2025hilora,zhang2024decomposing}. A critical research direction is how to optimize responses from LLMs by exploring prompts~\cite{liu2023pre}. Currently, the research community is actively refining evaluation methodologies, such as prompt engineering~\cite{guo2023connecting,pan2023plum,pryzant2023automatic}, to automatically enhance the accuracy of LLM responses and their alignment with user preferences. This process is collectively termed ``\textbf{LLM response identification.}''
In the offline phase, prompt-driven approaches generate diverse response styles and formats. However, the complexity and diversity of response styles ~\cite{zhang_harnessing_2024,li_stylechat_2024} pose significant challenges. Currently, LLM response evaluation relies on offline pointwise scoring, assigning independent scores to each response~\cite{liu2025inference}. This evaluation approach often requires generating numerous candidate responses and scoring them individually, which is time-consuming and computationally expensive due to the resource-intensive nature of LLM inference and evaluation metrics \cite{hu_online_2024, zhang2025mole}. For instance, evaluating 205 zero-shot prompts on 784 GSM8K questions using Mistral-7B requires 78 Nvidia A6000 GPU hours ~\cite{zhou_speeding_2025}. In contrast, online evaluation restricts the candidate response set, enabling efficient identification of optimal responses.

\begin{figure}[h]
   \centering
   \includegraphics[width=\linewidth]{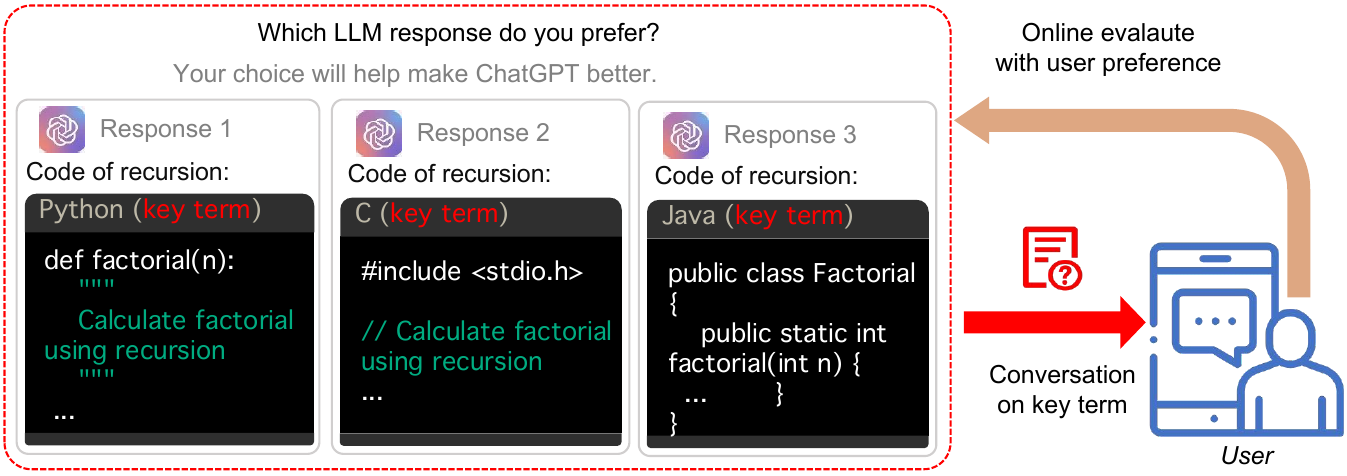}
   \caption{Evaluation aligned with online user feedback.}
   \label{fig:example}
\end{figure}
This study investigates the ``\textit{online}'' evaluation of a curated set of candidate responses to select the optimal one. An attractive feature of online evaluation is its capability to adapt to personalized user preferences, such as humorous or formal tones, while ensuring high response quality. Although prior work (e.g.,~\cite{rafailov2024direct}) has explored preference optimization, it is often limited to binary ``like/dislike" judgments. As preference demands grow increasingly diverse, LLM response identification must align with human preferences. To achieve alignment with human preferences, LLM response selection must prioritize both quality and user preference alignment. Traditional approaches rely on offline human quality control to build response pools with diverse styles, grouping responses by key terms (e.g., ``C/C++" or ``humorous tone")~\cite{Zhang-Conversational-WWW20,wu-2021-clustering-of-conversational,zhang2014addressing}, where these key terms can generalize to related responses for enabling preference inference. Our online frameworks take this further by allowing agents to directly query users about preferences on key terms, thereby advancing and improving the preference alignment process. For example, as shown in Fig. \ref{fig:example}, when ChatGPT presents different code styles of recursion, user online feedback on programming style can optimize subsequent responses, enhancing preference alignment accuracy. However, multi-device access through distributed agents (e.g., queries on the Poe AI platform~\cite{poe} across smartphones, tablets, and desktops) generates fragmented and heterogeneous preference data, complicating the alignment process.

To adaptively identify suitable responses from an initial offline response set and align with human preferences, we adopt an ``\textit{online contextual conversational bandit framework}'', a dynamic and interactive response selection mechanism. Specifically, the server treats each potential LLM response as an ``arm'' and employs multi-armed bandit (MAB) algorithms to actively evaluate and select the next ``arm'' for response identification. Using this contextual conversational bandit framework, the server can also query users about preferences for key terms (e.g., humorous or serious tone)~\cite{lei2020conversational,xie-2021-comparison-based}, accelerating preference learning and alignment. However, existing conversational bandit methods fall short in addressing the unique challenges of LLM response identification, which include:

\ding{182} Existing preference-aware bandit models, primarily used in recommendation systems~\cite{wu-2021-clustering-of-conversational,liu-2022-federated}, typically employ singular value decomposition (SVD)  to extract low-dimensional feature vectors. However, the complex semantic information in LLM responses results in high-dimensional feature spaces \cite{cheng2025tradingvectordatavector}, significantly increasing computational complexity.

 \ding{183} Most conversational bandits assume an infinite arm framework~\cite{liu2025leveraging,lei2020conversational}. In contrast, prompt engineering typically generates a large but finite set of responses as candidates. Current elimination-based contextual bandit algorithms can handle finite arm sets but rely on computationally intensive G-optimal designs~\cite{huang-2021-federated,lattimore-2020-bandit-algorithms,li2024fedconpe} to determine response selection distributions.

 \ding{184} Existing conversational bandit studies~\cite{Wang-2023-Efficient,Zhang-Conversational-WWW20,dai2024misspecification} use predefined functions to control conversational frequency, which typically follow a fixed sequence of engagements to initiate a specific number of conversations. But this fixed interaction model struggles to adapt to the dynamic needs of LLM response identification, failing to adaptively adjust the direction and frequency of key term conversations.
    
 \ding{185} As mentioned, the multi-device access inherent to LLMs significantly complicates conversational design. Existing conversational bandit methods are limited to single-agent scenarios and fail to address multi-device data access and processing, hindering efficient preference alignment. Conversely, existing multi-agent bandit studies multi-agent bandit studies~\cite{huang-2021-federated,lin-2023-federated,wang-2020-distributed-bandit} require local agents to upload user feedback or share identical arm sets, potentially leading to reduced flexibility and increased communication costs.

  This paper makes the following contributions:
  
$\bullet$ \textbf{Model Formulation}: We design a multi-agent framework for online evaluation and selection of LLM responses, enhancing offline evaluation methods by incorporating user preferences for alignment, and implementing a centralized mechanism to achieve decision consensus among agents.

$\bullet$ \textbf{Algorithm Design}: We propose \maco, multi-agent conversational online learning, comprising \texttt{MACO-A} (executed by local agents) and \texttt{MACO-S} (executed by cloud server). \texttt{MACO-A} employs an online elimination mechanism to quickly filter out low-quality responses, while \texttt{MACO-S} dynamically adjusts selection strategies based on aggregated preference data. Our adaptive preference mechanism triggers conversations only when necessary, thus aligning with human preferences more efficiently.
    
$\bullet$ \textbf{Theoretical Analysis}: We rigorously prove regret upper and lower bounds for our proposed \maco, demonstrating near-optimal performance. The innovative conversational mechanism maintains performance while controlling communication costs at $\mathcal{O}(d^2M\log T)$, eliminating the need for computationally intensive G-optimal designs used in traditional elimination-based bandits. 
    
$\bullet$ \textbf{Empirical Validation}: Extensive experiments are conducted using two text embedding models from Google and OpenAI, on open-source LLM Llama and real-world datasets. Results demonstrate that the \maco significantly outperforms baselines across varying response set sizes and agent counts. Moreover, \maco reduces time overhead substantially, while preserving performance.

\section{System Model}
\label{sec:models}
This section formulates the multi-agent conversational bandit for online LLM response evaluation and selection. 
 For any  vector $\vec{x}$ and a positive semi-definite matrix \(\vec{M}\), let \(\|\vec{x}\|_{\vec{M}}\coloneqq\sqrt{\vec{x}^\mathsf{T} \vec{M} \vec{x}}\). Denote the cardinality of a set \(\mathcal{A}\) as \(|\mathcal{A}|\) and \([z] \coloneqq \{1, \dots, z\}\) for \(\forall z \in \NN^{+}\).

\subsection{Online LLM Response Identification}
We define the set of local agents as \(\mM\) with $|\mM|=M$, which represent devices such as smartphones, laptops, and tablets. For any local agent \(m \in \mM\), the finite arm set of LLM responses is denoted as \(\mA_m\), which represents possible responses generated from various prompts. Given the heterogeneity of agents, different local agents may have different arm sets. 
As mentioned in Section \ref{sec:introduction}, traditional offline techniques (e.g., prompt engineering) can help to construct a set of initial responses, but due to the diversity of LLM outputs and user preferences, it is essential to online evaluate LLM responses and select the optimal one, despite having an offline initiatory set of LLM responses.
Our model adopts a time-slotted approach, denoted by discrete-time rounds \(\mathcal{T} = \{1, 2, 3, \ldots, T\}\), where each local agent selects one arm, i.e., LLM response, at each round \(t \in \mathcal{T}\).

\subsection{Multi-Agent User-Personalized Bandits}
  
 We consider a multi-agent conversational bandit setting involving \(M\) agents and a cloud server. At each round \(t \in \mathcal{T}\), a local agent \(m \in \mM\) selects an arm \(a_{m,t} \in \mathcal{A}_m\), which denotes one possible LLM response, and receives reward feedback \(r_{m,t}\) reflects the user satisfaction. 
 The user's preference for LLM responses is represented by an ``\textit{unknown}'' preference feature vector \(\vec{\theta}^{*} \in \mathbb{R}^d\), which all local agents aim to learn. For a local agent \(m \in \mM\), considering both the impact of the LLM response and the unknown user preference, the reward can be expressed as a linear combination with a noise term \(\eta_{m,t}\):
$r_{m,t} = \langle \vec{x}_{a_{m,t}}, \vec{\theta}^{*}_{t} \rangle + \eta_{m,t},$
where $\vec{x}_{a_{m,t}} \in \RR^d$ is the embedding feature vector of the response, to capture the textual information \cite{bubeck2023sparks,liu2023pre}. We will demonstrate the generalization of our model using two different open embedding approaches in Section \ref{sec:evaluation}.
Our objective is to design a policy that evaluates and selects arms (i.e., LLM responses) each round to minimize cumulative regret, defined as the difference between the cumulative rewards of our policy and the best \textit{unknown} policy across all local agents, tailored to personalized user preferences, which is defined as  \cite{Abbasi-Improved-2011,dai2024online,wu-2021-clustering-of-conversational}:
\begin{equation}\label{eq:regret}
    R_M(T) = \sum_{m=1}^{M} \sum_{t=1}^{T} \left( \vec{x}_{a_{m,t}^*}^\mathsf{T} \vec{\theta}^{*}_{t} - \vec{x}_{a_{m,t}}^\mathsf{T} \vec{\theta}^{*}_{t} \right),
  \end{equation}
where \(a_{m}^* \in \arg\max_{a \in \mathcal{A}_{m}} \vec{x}_a^\mathsf{T} \vec{\theta}^{*}\) denotes the optimal arm with the highest expected reward at local agent \(m \in \mM\). 

\subsection{Conversational Contextual Mechanism}

To adapt to and align with personalized user preferences, the cloud server also queries users' local agents for feedback, beyond selecting arms. However, directly using all responses risks \textit{information dispersion}, where ambiguous or overlapping groupings hinder clarity.
To counter this, as shown in Fig. \ref{fig:design}, we adopt the classical \textit{key terms} concept \cite{Zhang-Conversational-WWW20,wu-2021-clustering-of-conversational,zhang2014addressing,li_stylechat_2024} to group related arms under core concepts, e.g., ``C/C++'' or ``humorous tone''. Feedback on a key term applies only to its associated arms, enabling efficient preference inference with minimal interaction.
In our multi-agent framework, local agent \(m\) queries a user with key term \(k_{m} \in \mathcal{K}_{m} \subseteq \mathcal{K}\) based on prompt $y$. Feedback, given user preference and noise \(\widetilde{\eta}_{m,t}\), is modeled as: \(\widetilde{r}_{k_m,t} = \langle \tilde{\vec{x}}_{k_{m,t}}, \vec{\theta}^{*}_{t} \rangle + \widetilde{\eta}_{m,t}\), where  \(\tilde{\vec{x}}_k \in \RR^d\) is the embedding vector of key term $k$. Unlike prior conversational bandits \cite{Zhang-Conversational-WWW20,wu-2021-clustering-of-conversational,zhao-2022-knowledge-aware}, which use fixed conversation schedules (e.g., linear or logarithmic in round \(t\)) that may query unnecessarily, our algorithm (detailed in Section~\ref{sec:algorithms}) queries \textit{adaptively}, engaging users only when needed to refine preference estimates, improving efficiency and user experience (see Section~\ref{sec:theoretical-analysis} for detailed comparisons).

\begin{figure}[h]
   \centering
   \includegraphics[width=\linewidth]{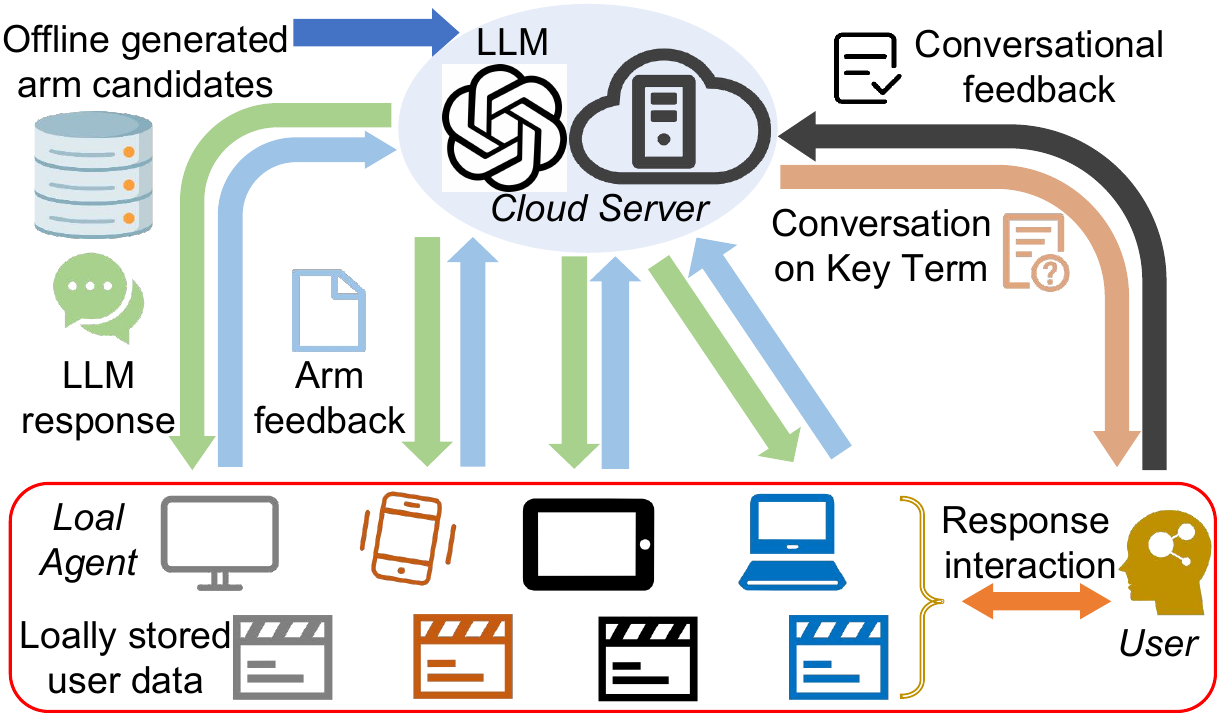}
   \caption{Multi-agent conversational bandit framework for online selecting LLM responses: \small{Local agents handle response selection (arms), while a central server manages conversation flow through key term selection. Server aggregates interaction data across multiple agents to accelerate user preference learning.}}
   \label{fig:design}
\end{figure}
\subsection{Distributed Communication Model}
 We consider a distributed model with \(M\) local agents and a cloud server, adopting a synchronous communication paradigm. In this setup, as shown in Fig. \ref{fig:design}, each local agent communicates with the cloud server by uploading and downloading data with negligible latency. Moreover, the local agents do not directly communicate with each other.
For simplicity, we focus on discrete-slot rounds solely for recording the selected arm. Querying key terms is interspersed with identifying LLM responses, allowing a key term to be queried and an arm to be pulled simultaneously. This aligns with the practical operations of conversational LLM systems. 
 Consistent with ~\cite{huang-2021-federated}, we define the standard communication cost as the cumulative count of scalar units transmitted between the cloud server and local agents, including both integers and real numbers.

\section{Algorithm Design}
\label{sec:algorithms}

  We present the design of multi-agent conversational online learning (\maco) algorithms.
  Then, we compare our design to the traditional phase elimination-based online learning algorithms. 
At a high level, multi-agent conversational bandits face the issue of agent heterogeneity: naively aggregating data across agents does not necessarily lead to better user preference estimation. To mitigate this, the central algorithmic challenge lies in how to \textit{guide exploration toward underrepresented directions in the feature space, ensuring accurate preference estimation across all relevant dimensions, with rigorous theoretical interpretability.}

\subsection{MACO Algorithm on Local Agent}
\label{sec:client-algorithm}

\begin{algorithm}[htb]
   \LinesNumbered
  \SetNlSty{textbf}{}{}
  \DontPrintSemicolon
  \small
  \SetKwComment{Comment}{$\triangleright$\ }{}
  \SetKwInput{KwInit}{Initialization}
  \KwIn{Local agent count \(M\), input dimension \(d\), arm pool size $A$, confidence value $\delta \in (0,1] $, phase $p=1$}
  \SetKwProg{Fn}{Function}{:}{}

  \While{\(T\) has not been reached}{
Calculate $\vec{M}_m^{p} = \sum_{a \in \mathcal{A}_m^{p}} \frac{1}{|\mathcal{A}_m^{p}|} \vec{x}_a \vec{x}_a^\mathsf{T}$ \label{line:information}; Diagonalize \(\vec{M}_m^{p}= \sum_{j=1}^{d} \lambda_{\vec{v}_j} \vec{v}_j \vec{v}_j^\mathsf{T}\)\label{line:diagonalization}\;
    Upload eigenvector \(\vec{v}_j\), if its corresponding eigenvalue satisfies \(\lambda_{\vec{v}_j} < h_{p} \coloneqq \frac{3}{4(1-2^{-2p})d}\)\label{line:check-eigenvalue}\;
    Download \(\mathcal{K}_m^{p}\) and \(\set{n_{m,k}^p}_{k \in \mathcal{K}_m^{p}}\) from the cloud server\label{line:download-key-terms}\;

\ForEach(\Comment*[f]{Conduct conversations}){\(k \in \mathcal{K}_m^{p}\) \label{line:conduct}}{
      Querying key term $k$ for \(n_{m,k}^p\) times;
      Receive rewards \(\set{\widetilde{r}_{k,t}}_{t \in \mathcal{\widetilde{T}}_{m,k}^{p}}\) from conversational feedback\;
    }
    
    \ForEach(\Comment*[f]{Pull arms}){\(a \in \mathcal{A}_m^{p}\)}{
    Set \(n_{m,a}^p=\left\lceil \frac{d}{2^{(-2p-1)}|\mathcal{A}_m^{p}|}\log \frac{2AM\log T}{\delta} \right\rceil\);
     Pull \(a\) for  $n_{m,a}^p$ times on the targeted LLM\;
      Receive rewards \(\set{r_{a,t}}_{t \in \mathcal{T}_{m,a}^{p}}\)  on the LLM response\label{line:finish}\;
    }
    
    \(
    \begin{aligned}
\text{Upload }\vec{G}_m^{p}& = \sum_{k \in \mathcal{K}_m^{p}} n_{m,k}^p \tilde{\vec{x}}_k \tilde{\vec{x}}_k^\mathsf{T} + \sum_{a \in \mathcal{A}_m^{p}} n_{m,a}^p\vec{x}_a \vec{x}_a^\mathsf{T}\text{, and}\\
      \vec{W}_m^{p} &= \sum_{t \in \bigcup_{k \in \mathcal{K}_m^{p}} \mathcal{\widetilde{T}}_{m,k}^{p}} \widetilde{r}_{k,t}\tilde{\vec{x}}_{k,t} + \sum_{t \in \bigcup_{a \in \mathcal{A}_m^{p}} \mathcal{T}_{m,a}^{p}} r_{a,t}\vec{x}_{a,t}  
    \end{aligned}
    \)\label{line:upload-data}\;
    Download \(\widehat{\vec{\theta}}_{p}\) from the cloud server;
  Update the active LLM response set $\mathcal{A}_m^{p+1}$ by eliminating sub-optimal LLM responses:
    \(\displaystyle 
    \mathcal{A}_m^{p+1} = \set{a \in \mathcal{A}_m^{p}: \max_{a^{\prime} \in \mathcal{A}_m^{p}} \inprod{\widehat{\vec{\theta}}_{p}}{\vec{x}_{a^{\prime}}-\vec{x}_a} \leq \frac{2^{-p+1}}{\sqrt{M}}}\) \label{line:elimination};
    \(p = p + 1\)\;
  }
  \caption{MACO on Local Agent (MACO-A)} \label{algo:client}
\end{algorithm}

As outlined in Algorithm~\ref{algo:client}, which is executed by the local agents and referred to as \texttt{MACO Agent (MACO-A)}, the online process of handling and updating information for LLM response evaluation within the multi-agent system operates as follows.  Define \(\mathcal{T}_{m,a}^{p}\) as the set of rounds where local agent \(m\) selects arm \(a\) in phase \(p\), \(\mathcal{\widetilde{T}}_{m,k}^{p}\) as the set of rounds when agent \(m\) conducts interaction on key term \(k\) in the same phase (like Fig. \ref{fig:example}), and $A \leq |\mathcal{A}|$ as the size of actually pulled arms from the LLM response set at each round.  During each phase \(p\), the local agent \(m \in \mM \) computes the \textit{information matrix} \(\vec{M}_m^p\) from its \textit{active arm set} \(\mathcal{A}_m^p\) (later updated in Line~\ref{line:elimination}). Specifically,  \(\vec{M}_m^p \coloneqq \sum_{a \in \mathcal{A}_m^p} \frac{1}{|\mathcal{A}_m^p|} \vec{x}_a \vec{x}_a^\mathsf{T}\), which captures principal directions in the feature space (Line~\ref{line:information}). The eigenvalue \(\lambda_{\vec{v}}\) of its eigenvector \(\vec{v}\) reflects variance along that direction, with larger values aiding precise estimation of \(\vec{\theta}^*\).
Local agent \(m\) then diagonalizes its information matrix \(\vec{M}_m^p = \sum_{j=1}^{d} \lambda_{\vec{v}_j} \vec{v}_j \vec{v}_j^\mathsf{T}\) to analyze all feature space directions (Line~\ref{line:diagonalization}). If an eigenvalue \(\lambda_{\vec{v}_j}\) falls below the threshold \(h_p \coloneqq \frac{3}{4(1-2^{-2p})d}\)  (determined by Lemma \ref{lemma:lower-bound-of-smallest-eigenvalue}), agent $m$ uploads the corresponding eigenvector to the cloud server (Line~\ref{line:check-eigenvalue}), addressing under-explored feature space regions to improve LLM response selection accuracy.
  
The cloud server processes uploaded data and sends local agent \(m\) a set of key terms \(\mathcal{K}_m^p\) with required repetition times \(\{n_{m,k}^p\}_{k \in \mathcal{K}_m^p}\) (Line~\ref{line:download-key-terms}). Agent \(m\) then queries these key terms and pulls arms as specified, ensuring thorough exploration of LLM responses. During this process, arm pulls and key term queries can be interleaved flexibly (Lines~\ref{line:conduct}-\ref{line:finish}), though shown sequentially for clarity. Agent \(m\) uploads data on pulled arms, key terms, and rewards, stored in matrices \(\vec{G}_m^p\) and \(\vec{W}_m^p\) (Line~\ref{line:upload-data}). Next, it downloads the updated preference estimate \(\widehat{\vec{\theta}}_p\) from the cloud server and updates its active arm set by eliminating suboptimal arms (Line~\ref{line:elimination}). This process enables agent \(m\) to adaptively and accurately evaluate and select LLM responses tailored to user preferences while minimizing external data sharing by only uploading aggregated data (\(\vec{G}_m^p\), \(\vec{W}_m^p\)) to the cloud server.

\subsection{MACO Algorithm on Cloud Server}
\label{sec:server-algorithm}
\begin{algorithm}[thb]
   \LinesNumbered
  \SetNlSty{textbf}{}{}
  \DontPrintSemicolon
  \small
  \SetKwInput{KwInit}{Initialization}
  \KwIn{Key term set \(\mathcal{K}\), coverage parameter \(\beta\) in Condition~\ref{cond:key-term-richness}.}
  \KwInit{Let \(p=1, \vec{G}=\vec{0}, \vec{W}=\vec{0}\)}
  \SetKwFunction{SupportExploration}{SupportExploration}
  \SetKwProg{Fn}{Function}{:}{}

  \While{\(T\) has not been reached}{
    \ForEach{\(m \in \mM\)}{
Receive all eigenvectors uploaded by local agent \(m\), and denote this set as \(\mathcal{S}_m\)\;
Initialize the set of key terms at phase \(p\)  as \(\mathcal{K}_m^{p} = \emptyset\)\;

      \ForEach{\(\vec{v}_j \in \mathcal{S}_m\)}{
        \(k = \argmax_{i \in \mathcal{K}} \tilde{\vec{x}}_i^\mathsf{T}\vec{v}_j\), \(\mathcal{K}_m^{p} = \mathcal{K}_m^{p} \cup \set{k}\) ;
        \(n_{m,k}^p = \left\lceil \frac{\frac{3}{2(1-2^{-2p})}-2d\lambda_{\vec{v}_j}}{\beta^2 2^{-2p}}\log \frac{2AM\log T}{\delta} \right\rceil\)\label{line:find-key-term}\;
      }
      Send \(\mathcal{K}_m^{p}\) and \(\set{n_{m,k}^p}_{\vec{k} \in \mathcal{K}_m^{p}}\) to local agent $m$\label{line:send-back-key-terms};
      Receive \(\vec{G}_m^{p}\) and \(\vec{W}_m^{p}\) from local agent $m$\;
    }
    \(\vec{G} = \sum_{p \in [p]} \sum_{m \in \mM} \vec{G}_m^p,\ \
    \vec{W} = \sum_{p \in [p]}\sum_{m \in \mM} \vec{W}_m^p\) \label{line:aggregate-data}\;
    Broadcast \(\widehat{\vec{\theta}}_{p} = \vec{G}^{-1} \vec{W}\) to all local agents \label{line:estimate-theta};
    \(p = p + 1\)\;
  }
  \caption{MACO on Cloud Server (MACO-S)} \label{algo:server}
\end{algorithm}

The \maco\ algorithm's cloud server component, \texttt{MACO Server (MACO-S)}, tackles the challenge of local agent heterogeneity in the multi-agent conversational bandits model, as noted in Section \ref{sec:introduction}. This diversity can hinder effective data aggregation, risking suboptimal estimation of the user preference \(\vec{\theta}^{*}\). To counter this, the cloud server strategically uses key terms to probe and enhance information in underrepresented feature space directions, thus improving the accuracy of the estimation process.

In Algorithm~\ref{algo:server}, the cloud server receives eigenvectors from local agents indicating under-explored directions in the LLM response space (Line~\ref{line:find-key-term}). It selects key terms \(k \in \mathcal{K}\) by maximizing the inner product with these directions and assigns repetition times \(n_{m,k}^p\), sending them to local agents (Line~\ref{line:send-back-key-terms}). This enables targeted exploration of LLM responses related to these key terms. The server then aggregates data from all agents and estimates the preference parameter \(\vec{\theta}^{*}\) using linear regression, reducing uncertainty and improving tailored response predictions (Lines~\ref{line:aggregate-data}-\ref{line:estimate-theta}). To ensure invertibility, especially for large \(d\), \(\vec{G}\) can be initialized as an identity matrix.

\subsection{Advantages over Phase Elimination Bandit}

As noted in Section \ref{sec:introduction}, LLM responses requiring online evaluation are finite, making phase elimination-based linear bandit algorithms (\texttt{PE-Lin}) preferable over classical conversational bandit frameworks \cite{Zhang-Conversational-WWW20} due to better performance guarantees for finite arm sets. We enhance \texttt{PE-Lin} \cite{lattimore-2020-bandit-algorithms}, where a learning agent estimates the preference vector \(\vec{\theta}^{*}\) using least squares with \emph{G-optimal design} to minimize prediction variance. This design optimizes a probability distribution \(\pi: \mathcal{X} \to [0,1]\) over arm feature vectors \(\mathcal{X} \subset \RR^d\), satisfying:
\begin{equation}\label{eq:G-optimal}
\begin{aligned}
\sum_{\vec{x} \in \mathcal{X}} \pi(\vec{x}) &= 1, \quad \vec{M}_m^p(\pi) = \sum_{\vec{x} \in \mathcal{X}} \pi(\vec{x}) \vec{x} \vec{x}^\mathsf{T}, \\
g(\pi) &= \max_{\vec{x} \in \mathcal{X}} \|\vec{x}\|_{\vec{M}(\pi)^{-1}}^2=d.
\end{aligned}
\end{equation}
The agent plays arms per \(\pi\), estimates \(\vec{\theta}^{*}\), and eliminates inferior arms. Computing \emph{G-optimal design} in multi-agent settings is inefficient \cite{huang-2021-federated}. Our algorithm, \maco, avoids this by leveraging multi-agent heterogeneity and adaptive conversational mechanisms, reducing computation costs. Running \texttt{PE-Lin} independently on each agent with server aggregation yields a regret bound of \(\mathcal{\widetilde{O}}(M \sqrt{dT})\), equivalent to no communication. Our approach improves this to \(\mathcal{\widetilde{O}}(\sqrt{dMT})\) by enabling information sharing through lightweight conversations, as we will analyze in Section~\ref{sec:theoretical-analysis}.

\section{Performance Analysis}
\label{sec:theoretical-analysis}
This section presents the theoretical results of \maco: regret, communication costs, and conversation frequency.

Following common practices in~\cite{xie-2021-comparison-based,dai2024misspecification}, we assume for any arm $a$ and key term $k$, \(\|\vec{x}_a\| = \|\tilde{\vec{x}}_k\| = 1\). The length of preference vector \(\vec{\theta}^{*}\) is bounded by 1, and the noise terms \(\eta_{m,t}\) and \(\widetilde{\eta}_{m,t}\) are  1-subgaussian. We first present a ``\textit{technical condition}'' that addresses general issues related to feature space coverage.

\begin{condition}[Feature Space Coverage]\label{cond:key-term-richness}
 We say that a key term set \(\mathcal{K}\) as \emph{sufficiently rich} for covering the feature space if, for any unit vector \(\vec{v} \in \RR^d\), there exists a key term \(k \in \mathcal{K}\) such that its feature vector $\tilde{\vec{x}}_k$ satisfies \(\tilde{\vec{x}}_k^\mathsf{T}\vec{v} \geq \beta\), where\(\beta \in (0,1]\) is a coverage parameter close to 1. 
\end{condition}

\begin{remark}
Condition~\ref{cond:key-term-richness} ensures key terms in \(\mathcal{K}\) are well-distributed across the feature space, enabling effective uncertainty reduction for each agent. This is satisfied if \(\mathcal{K}\) contains an orthonormal basis of \(\RR^d\). It allows for bypassing the \emph{G-optimal design} used in traditional elimination-based algorithms to minimize prediction variance.
\end{remark}

For sufficiently rich key term sets, based on Condition \ref{cond:key-term-richness}, we provide the following theorems.

\begin{restatable}[Regret Bounds]{theorem}{restateregret}\label{thm:regret}
 We have the following upper and lower regret bounds:
  \begin{enumerate}[leftmargin=*]
    \item \textbf{Upper Bound:} With probability at least \(1-\delta\), the  regret is bounded by \(\mathcal{O}(\sqrt{dMT\log \frac{AM\log T}{\delta}})\).
        \item \textbf{Lower Bound:} For any policy that selects at most one key term per round, there exists an instance where the policy incurs an expected regret of at least \(\Omega(\sqrt{dMT})\). 
  \end{enumerate}
\end{restatable}

\begin{remark}\label{remark:regret}
Theorem~\ref{thm:regret} highlights key performance insights for our approach: 1) For \(M=1\), the problem reduces to single-agent conversational bandits, with regret \(\mathcal{\widetilde{O}}(\sqrt{dT})\), outperforming prior bounds of \(\mathcal{\widetilde{O}}(d\sqrt{T})\) \cite{Wang-2023-Efficient,Zhang-Conversational-WWW20} by using phase elimination on finite arm sets, especially for high-dimensional LLM feature vectors. 2) In multi-agent settings, our upper bound matches the results of \cite{huang-2021-federated,li2024fedconpe} while avoiding computationally expensive \emph{G-optimal design}, accelerating the online process. Moreover, unlike \cite{wang-2020-distributed-bandit} which assumes all local agents share the same arm set, agents here have heterogeneous arm sets.
3) Collectively, the regret upper and lower bounds indicate \maco is minimax optimal up to a logarithmic factor, showing its tightness.
\end{remark}

\begin{restatable}[Communication Cost]{theorem}{restatecommunication}\label{thm:communication}
The communication cost scales in  \(\mathcal{O}(\log T)\), independent of arm pool size \(A\).
\end{restatable}

\begin{remark}\label{remark:communication}
The arm pool size $A$ can reach thousands due to the diversity of LLM-generated responses. In contrast to \cite{huang-2021-federated}, where the communication cost scales as $\mathcal{O}(A \log T)$, our approach significantly reduces overhead. Specifically, local agents process data independently and transmit only aggregated results to the cloud server, avoiding the need to upload full arm sets of size $\mathcal{O}(A)$.
\end{remark}

\begin{restatable}[Bound on Conversation Frequency]{theorem}{restateconversation}\label{thm:conversation}
  For any local agent \(m \in \mM\) during phase \(p\), let \(\gamma = \lambda_{\text{min}}(\vec{M}_m^{p})\), where \(\lambda_{\text{min}}\) denotes the smallest eigenvalue, we have: 1) If \(\gamma \geq h_p\), no conversations will be initiated. 2) If \(\gamma < h_p\), the fraction of conversations relative to the total phase length is capped at \(\beta^{-2}(\frac{3}{4(1-2^{-2p})}-d\gamma)\).
\end{restatable}

\begin{remark}\label{remark:conversation}
\maco employs an adaptive conversation method, unlike deterministic schedules \(b(t)\) (e.g., linear or logarithmic) used in prior conversational bandit studies \cite{Zhang-Conversational-WWW20,Wang-2023-Efficient,xie2021comparison}. Fixed-interval conversations can be inefficient, especially when user preferences are already well understood. In contrast, our model dynamically adjusts based on gaps in user preferences, enabling a more realistic interaction paradigm.
\end{remark}

\section{Performance Evaluation}
\label{sec:evaluation}
\begin{figure*}
    \centering
    \includegraphics[width=\linewidth]{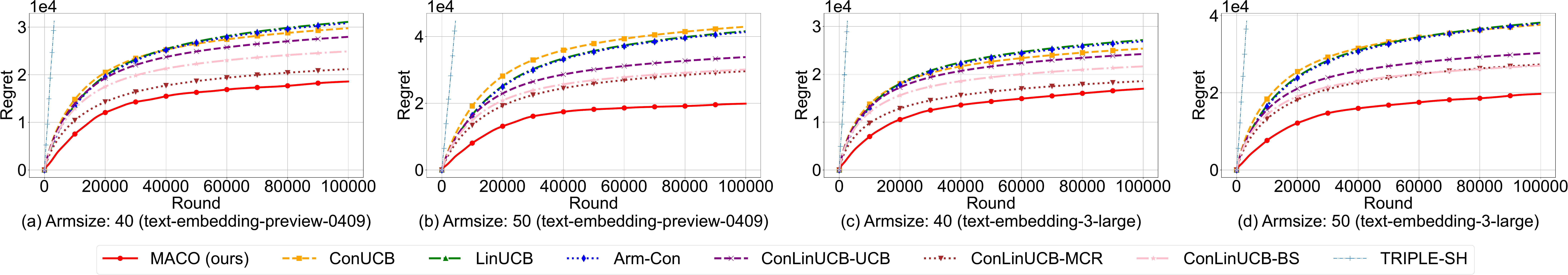}
     
    \caption{Regret on embedding models from Google and OpenAI across different arm pool sizes $A$.}
    \label{fig:result_simple}
        
\end{figure*}

\subsection{Experimental Settings}
\label{subsec:experimental-settings}

\textbf{Embedding Models.} We demonstrate our framework's generalization capabilities using different open embedding models from Google and OpenAI, which generate the embedding feature vector $\vec{x}_{a} \in \RR^d$ for the corresponding arm $a$ (i.e., response) to capture text information:

\ding{182} \textbf{Text-embedding-preview-0409:} Google's advanced embedding model, which streamlines synthetic training data creation by generating queries and task descriptions \cite{lee2024gecko}. 
 \ding{183}\textbf{Text-embedding-3-large:} OpenAI's generation embedding model, which surpasses its predecessor, though its technical details remain undisclosed \cite{muennighoff2023mtebmassivetextembedding}.
 \ding{184}\textbf{Text-embedding-multilingual-002:} To evaluate the algorithm's generalization across different languages, we also employ a multilingual embedding model of Google \cite{lee2024gecko}.
 
\textbf{Dataset Preprocessing.} We implement two response settings using the embedding models \ding{182} and \ding{183}, leveraging real-world datasets and the open-source Llama-3-8B-Instruct LLM \cite{Ollama, sahoo2024systematicsurveypromptengineering}. The first configuration generates 510 unique response arms from combinations of 13 key terms (e.g., ``humorous,'' ``helpful'') \cite{köpf2023openassistantconversationsdemocratizing}, while the second yields 455 arms by pairing five questions with two keyword styles. Both settings optimize response selection through cosine similarity between user preference and response feature vectors \cite{reimers-gurevych-2019-sentence}. Details are provided in \Cref{app:experiment}. 

In addition to Llama-generated responses, we evaluate \maco on the StyleEval \cite{li_stylechat_2024} and Gretel Multilingual \cite{gretel-synthetic-multilingual-llm-prompts-2024} datasets, which encompass diverse response styles (e.g., ``polite,'' ``romantic'') and multilingual key-term variations for conversational LLM interfaces. The original multilingual dataset, limited to six languages, was expanded to 20 languages using GPT-4o to support larger-scale experiments \cite{openai_gpt35turbo}. As mentioned, the embedding model \ding{184} is employed for multilingual support.
A sample LLM response conversation is included in Fig. \ref{fig:prompt}  for reference.

\begin{figure}[!ht]
    \centering  
    \includegraphics[width=\linewidth]{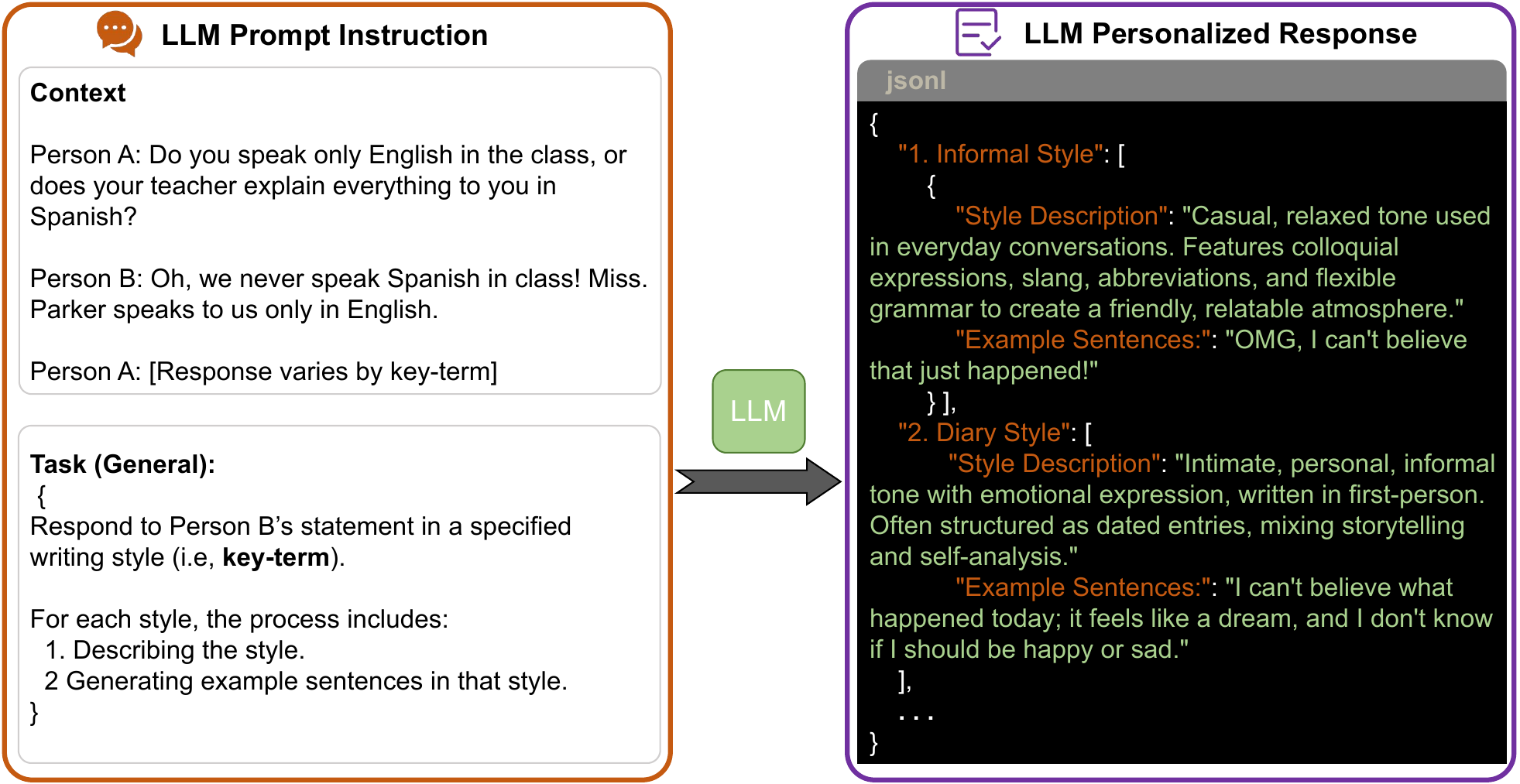}
    \caption{A sample LLM response conversation for user-aligned evaluation.}
    \label{fig:prompt}

\end{figure}

   \textbf{Comparison Algorithms.} Results are averaged over five trials on a Linux Ubuntu machine (kernel 6.5.0) with a 5.40 GHz 13th Gen Intel Core i7-13700KF CPU and 32GB RAM. We set coverage parameter $\beta = 1$ and confidence parameter  
 $\delta = 0.1$, with an ablation study in \Cref{app:experiment}.
Baselines, run independently on local agents, include:
 1) \texttt{TRIPLE-SH} \cite{shi2024best}: Select optimal prompts for LLMs by adaptively eliminating arms with poor performance, where we directly set each arm as the corresponding LLM response.
  2) \texttt{LinUCB}~\cite{Abbasi-Improved-2011}: Online select arms and estimate user preference for \emph{infinite} arm sets, excluding the conversational setting.
 3) \texttt{Arm-Con}~\cite{christakopoulou-2016-towards-conversational}: Initiate conversations on user preference about arms, and use \texttt{LinUCB} for arm selection.
4) \texttt{ConUCB}~\cite{Zhang-Conversational-WWW20}: Query key terms if conversations are allowed and utilize conversational feedback to accelerate learning.
5) \texttt{ConLinUCB}~\cite{Wang-2023-Efficient}: Includes \texttt{ConLinUCB-BS} (uses barycentric spanner for conversations), \texttt{ConLinUCB-MCR} (selects terms with largest confidence radius), and \texttt{ConLinUCB-UCB} (uses LinUCB for term selection).\footnote{Although practical scenarios often have pre-existing information such as user preferences, our study focuses on the online learning process without assuming any \textit{prior knowledge}.}




\subsection{Evaluation Results}
\label{sec:evaluation-results}

\begin{figure*}
    \centering
    \includegraphics[width=\linewidth]{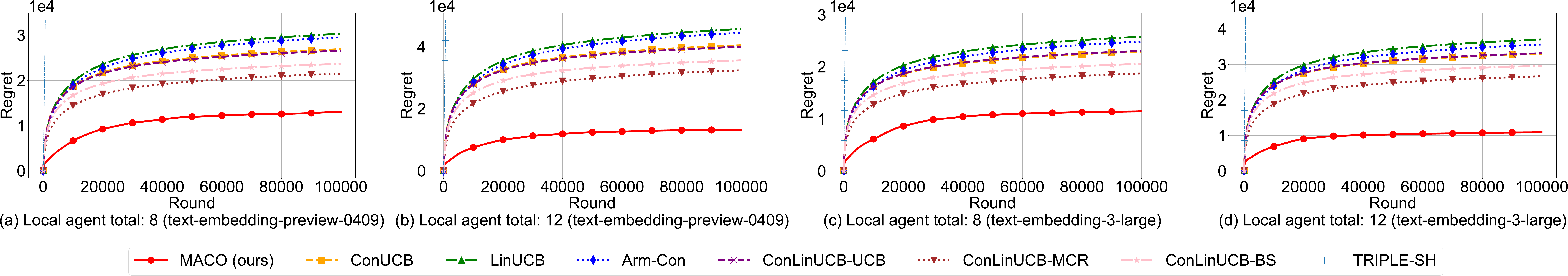}
     
    \caption{Regret on embedding models from Google and OpenAI across different agent counts $M$.}
    \label{fig:result_complex}
    
\end{figure*}

\begin{table*}[t]
    \centering
    \caption{Cumulative regret for algorithms under two embeddings on the StyleEval dataset.}\label{tab:realdata1}
    \resizebox{\linewidth}{!}{
    \begin{tabular}{lcccccccc}
        \toprule
        & \multicolumn{4}{c}{Text Embedding Model of Google} & \multicolumn{4}{c}{Text Embedding Model of OpenAI} \\
        \cmidrule(lr){2-5} \cmidrule(lr){6-9}
        Algorithm & $M=4$ & $M=8$ & $M=12$ & $M=16$ & $M=4$ & $M=8$ & $M=12$ & $M=16$ \\
        \midrule
        \small{\texttt{TRIPLE-SH}} 
            & 5847.31$\pm$19.80 & 11478.73$\pm$23.16 & 17151.16$\pm$36.35 & 22673.76$\pm$21.64 
            & 7736.87$\pm$36.32  & 15302.63$\pm$62.52 & 22777.67$\pm$62.34 & 30138.45$\pm$25.60 \\
        \texttt{LinUCB}   
            & 495.67$\pm$10.87   & 1007.75$\pm$20.83   & 1510.77$\pm$31.77   & 2025.16$\pm$23.69   
            & 401.16$\pm$14.43    & 808.91$\pm$30.78    & 1214.77$\pm$30.88    & 1625.90$\pm$22.20  \\
        \texttt{Arm-Con}  
            & 425.89$\pm$12.13   & 871.32$\pm$15.63    & 1295.84$\pm$32.97   & 1742.96$\pm$53.81   
            & 319.82$\pm$14.35    & 645.34$\pm$30.84    & 953.95$\pm$38.38     & 1297.47$\pm$42.48  \\
        \texttt{ConUCB}   
            & 237.62$\pm$5.81    & 479.36$\pm$9.91     & 721.05$\pm$18.88    & 960.33$\pm$14.34    
            & 190.36$\pm$7.34     & 382.64$\pm$15.52    & 583.57$\pm$17.97     & 779.50$\pm$10.44   \\
        \tiny{\texttt{ConLinUCB-BS}} 
            & 991.73$\pm$27.14   & 2000.84$\pm$45.84   & 3005.28$\pm$67.18   & 4011.74$\pm$16.24   
            & 781.52$\pm$27.47    & 1583.90$\pm$58.27   & 2385.23$\pm$56.67    & 3177.74$\pm$49.60  \\
        \tiny{\texttt{ConLinUCB-MCR}}
            & 439.70$\pm$13.92   & 895.52$\pm$15.70    & 1329.91$\pm$31.22   & 1793.32$\pm$16.30   
            & 382.10$\pm$15.50    & 769.74$\pm$29.30    & 1146.81$\pm$35.96    & 1536.44$\pm$34.82  \\
        \tiny{\texttt{ConLinUCB-UCB}}
            & 1145.26$\pm$21.32  & 2313.41$\pm$33.67   & 3463.31$\pm$52.00   & 4634.67$\pm$43.57   
            & 937.34$\pm$31.69    & 1889.40$\pm$50.69   & 2836.83$\pm$46.72    & 3805.12$\pm$39.17  \\
        \texttt{MACO} (Ours)     
          & \textbf{39.04$\pm$1.58} & \textbf{77.11$\pm$2.01} & \textbf{115.48$\pm$2.78} & \textbf{153.12$\pm$1.55} 
            & \textbf{55.03$\pm$1.98} & \textbf{108.50$\pm$2.98} & \textbf{162.42$\pm$3.26} & \textbf{215.71$\pm$3.37} \\
        \bottomrule
    \end{tabular}
    }
    
\end{table*}

  \textbf{Regret Across Different Arm Pool Sizes.}
We compare \maco's cumulative regret against seven baseline algorithms in the first setting with \(M=4\) local agents, using two embedding models. We vary arm pool sizes (\(A=40, 50\)), randomly selecting \(A\) arms from \(\mathcal{A}\) per local agent. Fig. \ref{fig:result_simple} shows that non-conversational algorithms (\texttt{LinUCB}, \texttt{Arm-Con}) perform worst, while \maco outperforms all baselines, improving at least 8.29\% over the best baseline, \texttt{ConLinUCB-MCR}. This stems from \maco's multi-agent framework, where the cloud server aggregates local agent data to better estimate user preferences. Increasing \(A\) has minimal impact on \maco's regret, supporting Theorem~\ref{thm:regret}, which indicates regret grows at a square-root logarithmic rate with \(A\).

  \textbf{Regret Across Different Numbers of Local Agents.}
We evaluate regret in the second setting with arm pool size \(A=40\), using the same embedding models, and varying the number of local agents (\(M=8, 12\)). Larger \(M\) reflects practical platforms grouping similar users to share learning, so we test \maco's performance with increased \(M\). Fig. \ref{fig:result_complex} shows that without a multi-agent framework, baseline algorithms' regrets grow linearly with \(M\), following \(\mathcal{\widetilde{O}}(dM\sqrt{T})\). In contrast, \maco leverages aggregated data from all agents, scaling regret as \(\mathcal{\widetilde{O}}(\sqrt{dMT})\), significantly reducing regret growth. This highlights \maco's effective multi-agent approach for online LLM response identification. Fig. \ref{fig:result_diff_clients}  further illustrates this trend.

\begin{figure}[!ht]
    \centering  
    
    \includegraphics[width=\linewidth]{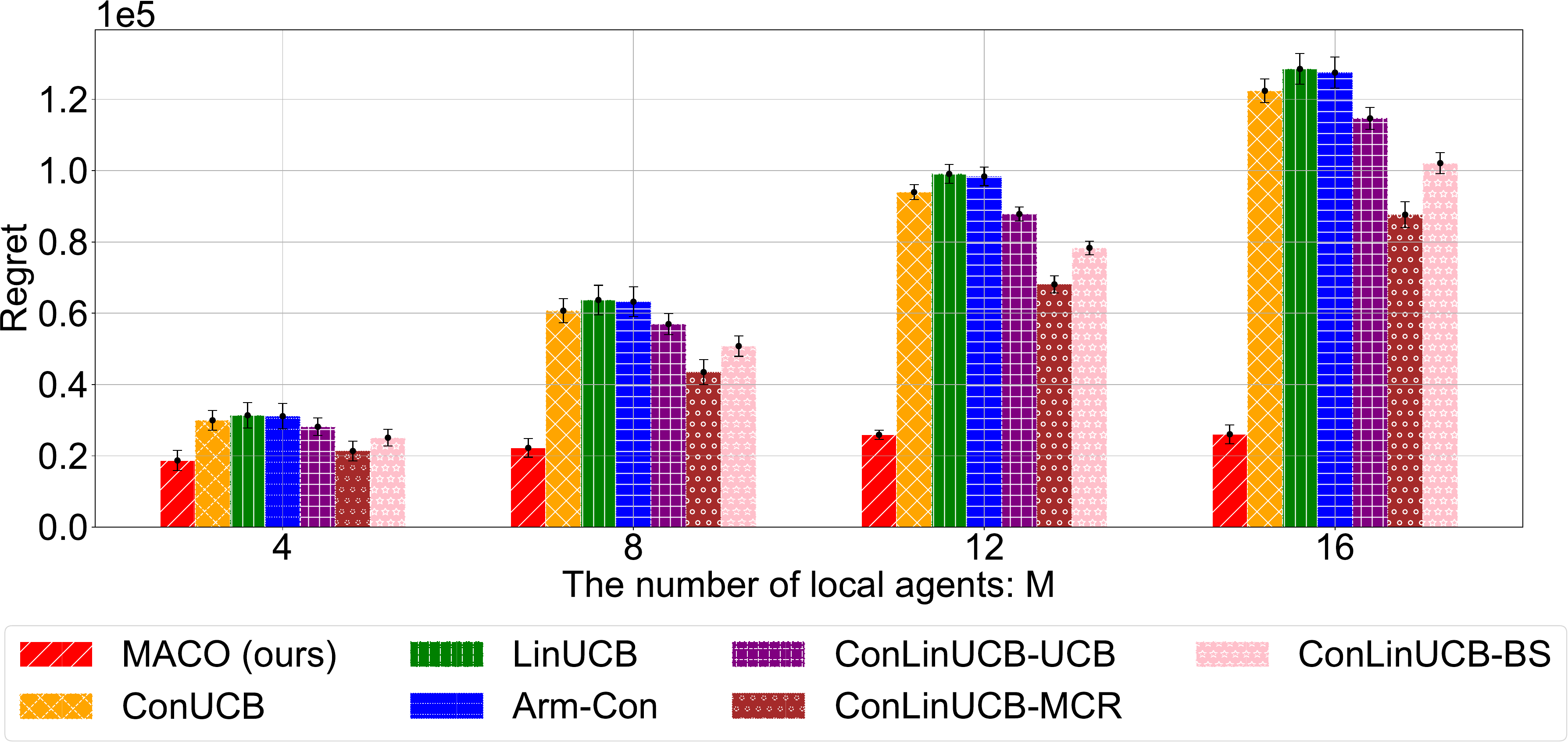}
     
    \caption{Regret under various numbers of local agents.}
    \label{fig:result_diff_clients}
   
\end{figure}

\begin{table}[h]
    \centering
    \caption{Execution time and reward on different settings.}
    \label{tab:combined}
    \resizebox{1\columnwidth}{!}{
    \begin{tabular}{|l|cc|cc|cc|}
    \hline
    \multirow{2}{*}{\makebox[0.155\textwidth][c]{\diagbox{\tiny{Setting}}{\tiny{Algorithm}}}} & \multicolumn{2}{c|}{\textbf{MACO (w/o G)}} & \multicolumn{2}{c|}{MACO (w/G)} & \multicolumn{2}{c|}{ConLinUCB-BS} \\ \cline{2-7}
    & Time (s) & Reward & Time (s) & Reward & Time (s) & Reward \\ \hline
    Setting (a) & 2.576 & 61.849 & 9.766 & 61.847 & 18.124 & 59.811 \\ 
    Setting (b) & 2.546 & 61.605 & 14.272 & 61.591 & 18.056 & 59.663 \\ 
    Setting (c) & 2.576 & 47.405 & 6.369 & 47.381 & 17.926 & 46.104 \\ 
    Setting (d) & 2.661 & 41.770 & 6.270 & 41.858 & 17.919 & 40.720 \\ \hline
    \end{tabular}}
\end{table}

\textbf{Comparison of Overhead on Execution Time.} 
We compare the execution time (s) of our algorithm, \texttt{MACO w/o G} for emphasis, against \texttt{ConLinUCB-BS} (noted as the fastest in \cite{Wang-2023-Efficient}) and \texttt{MACO w/G} (using traditional \textit{G-optimal design}) under \(T=5000\), 6 phases (\(A=40, M=4\)). The results of the two embedding models under the two response settings are labeled Settings (a), (b), (c), and (d). \Cref{tab:combined} (full version in \Cref{app:experiment}) shows \texttt{MACO w/o G} significantly reduces execution time by avoiding \textit{G-optimal design} and leveraging multi-agent data aggregation to speed up learning. Additionally, \texttt{MACO w/o G} maintains the same average reward as \texttt{MACO w/G}, confirming that our conversational approach preserves performance while replacing \textit{G-optimal design} with a more efficient, conversation-based design.

\textbf{Performance on Real-world Datasets.}
We evaluate \maco's cumulative regret against seven baseline algorithms on the StyleEval \cite{li_stylechat_2024} and Gretel Multilingual \cite{gretel-synthetic-multilingual-llm-prompts-2024} datasets with $T=1000$, using Google and OpenAI text embedding models across varying numbers of local agents (\(M=4, 8, 12, 16\)). \Cref{tab:realdata1} presents results on the StyleEval dataset, where \maco consistently achieves the lowest regret, outperforming the best baseline (\texttt{ConUCB}) by at least 79.56\% for \(M=4\) (Google embeddings) and 71.08\% (OpenAI embeddings). Similarly, the results on the Gretel Multilingual dataset (with details in \Cref{tab:realdata2}), with \maco reducing regret by at least 80.07\% (Google embeddings) and 69.23\% (OpenAI embeddings) compared to \texttt{ConUCB} for \(M=4\). Across both datasets, \maco's regret scales sublinearly with \(M\), validating its robust multi-agent framework for online  LLM response evaluation and selection even in diverse and multilingual settings.

\section{Related Work}
\label{sec:related-work}
Research on prompt learning for automatically generating suitable LLM responses has made significant progress \cite{guo2023connecting,zhang2023auto}. For example, \citet{li_stylechat_2024} designs StyleChat for stylized response generation. However, offline generating methods face challenges like ``data drift,'' emphasizing the need for online approaches to optimize LLM responses \cite{ekya,chen2023frugalgpt}. \citet{xia2024llm} introduces an online non-stationary bandit method across different LLMs. \citet{shi2024best} proposes an online budget-limited LLM response optimization using various prompts. And \citet{dai2024cost} focuses on response identification over multiple LLM coordination. Nevertheless, these studies ignore the impact of user preferences and the natural multi-agent setting in LLM response identification.

  Bandits tackle the exploitation-exploration tradeoff of online decision-making problems \cite{Abbasi-Improved-2011}. Based on this, conversational contextual linear bandits, introduced by \cite{Zhang-Conversational-WWW20}, allow the cloud server to obtain user feedback on key terms to elicit preferences, in addition to arm selection. Later studies introduce clustering to avoid labeling efforts \cite{wu-2021-clustering-of-conversational}, integrate knowledge graphs for term selection \cite{zhao-2022-knowledge-aware}, and compute the barycentric spanner as an efficient exploration basis \cite{Wang-2023-Efficient}. 
 Regarding the multi-agent bandit setting under finite arm sets, \citet{wang-2020-distributed-bandit} assumes homogeneous arm sets, \citet{huang-2021-federated}   requires the local agents to upload arm sets, increasing costs and privacy concerns, and \cite{li2024fedconpe} utilizes the computationally intensive \textit{G-optimal design}. Unlike existing works, we are the first to extend conversational bandits to multi-agent settings for online LLM response adaptation with reduced computation resources.

\section{Conclusion and Future Work}
\label{sec:conclusion}
This paper presents \(\text{\maco}\), a novel multi-agent conversational bandit framework designed to online evaluate and select optimal responses from LLMs while minimizing cumulative regret and aligning with user preferences. Our \texttt{MACO} framework consists of local agents (\texttt{MACO-A}) that adaptively manage conversations and response selection, and a cloud server (\texttt{MACO-S}) that aggregates data to learn user preferences efficiently. We have proved that \(\text{\maco}\) achieves optimal regret bounds, reduces conversations, and enhances computational efficiency. Our extensive evaluations, using open LLMs such as Llama and GPT-4o, confirm that our approach significantly improves performance over traditional methods on real-world datasets. 
Future work could explore clustering similar user preferences, integrating privacy-preserving techniques, and extending the framework beyond linear reward models.

\section*{Ethics Statement}
This research raises no ethical concerns, with no involvement of human subjects or potential for harm.

\section*{Acknowledgement}
The work of John C.S. Lui was supported in part by the RGC GRF-14215722.

\clearpage

\appendix

\section{Proof Appendix}\label{app:proof}

\subsection{Facts}
\label{sec:preliminaries}
  We first introduce some well-known facts without proofs.

  \begin{fact}[Subgaussian random variables]\label{lemma:subgaussian}
    Suppose that random variables \(X\) is \(\sigma\)-subgaussian, \(X_1\) and \(X_2\) are independent and \(\sigma_1\) and \(\sigma_2\)-subgaussian, respectively, then
    \begin{enumerate}
      \item For any \(\varepsilon>0\), \(\Pr\left[X\geq \varepsilon\right] \leq \exp\left(-\frac{\varepsilon^2}{2\sigma^2}\right)\).
      \item \(X_1 + X_2\) is \(\sqrt{\sigma_1^2 + \sigma_2^2}\)-subgaussian.
    \end{enumerate}
  \end{fact}

  \begin{fact}[\protect\citeauthor{bretagnolle-huber-1978}]\label{lemma:bretagnolle-huber}
    Let \(P\) and \(Q\) be probability measures on the same measurable space \((\Omega, \mathcal{F})\), and let \(A \in \mathcal{F}\) be an arbitrary event. Then,
    \[P(A) + Q(A^c) \geq \frac{1}{2} \exp(-D(P \parallel Q)),\]
    where \(D(P \parallel Q)=\int_{\Omega}\log\left(\odv{P}{Q}\right)\odif{P} = \E_P\left[\log\odv{P}{Q}\right]\) is the KL divergence between \(P\) and \(Q\). \(A^c = \Omega \setminus A\) is the complement of \(A\).
  \end{fact}

  \begin{fact}[KL divergence between Gaussian distributions]\label{lemma:kl-divergence-of-gaussian}
    If \(P \sim \mathcal{N}(\mu_1, \sigma^2)\) and \(Q \sim \mathcal{N}(\mu_2, \sigma^2)\), then
    \[D(P \parallel Q)=\frac{(\mu_1-\mu_2)^2}{2\sigma^2}.\]
  \end{fact}

\subsection{Proof of Regret Upper Bound in Theorem \ref{thm:regret}}\label{subsec:analysis}

\begin{proof}

We now provide an analysis of the upper bound in Theorem \ref{thm:regret}. Proofs for other theorems can be found in~\Cref{app:thm-regret}, ~\Cref{sec:proof-communication}, ~\Cref{sec:proof-conversation}. Below, we present two critical lemmas related to the design of our multi-agent conversational bandit algorithm. \Cref{lemma:lower-bound-of-smallest-eigenvalue} guarantees that for any local agent \(m\), the smallest eigenvalue of the information matrix, adjusted for conversational feedback, remains above \(h_p\). This supports the design of line \ref{line:check-eigenvalue} in Algorithm \ref{algo:client}. \Cref{lemma:bad-event} ensures that the algorithm operates within established error limits, which is essential for online, reliable LLM response identification.

\begin{lemma}[Stability of the Information Matrix]\label{lemma:lower-bound-of-smallest-eigenvalue}
  For any local agent \(m \in \mM\) during phase \(p \), we have
  $\lambda_{\text{min}}(\vec{M}_m^{p'}) \geq h_p,$
  where $\vec{M}_m^{p'} \coloneqq \vec{M}_m^{p} + \sum_{k \in \mathcal{K}_m^{p}} \frac{h_p-\lambda}{\beta^2}\tilde{\vec{x}}_k \tilde{\vec{x}}_k^\mathsf{T}$.
\end{lemma}

\begin{proof}
  Using the eigenvectors as an orthonormal basis, for any \(j \in [d]\), any key term's $k$ feature vector can be expressed as \(\tilde{\vec{x}}_k = \sum_{i=1}^{d} c_i \vec{v}_i = \sum_{i=1, i\neq j}^d c_i \vec{v}_i + c_j \vec{v}_j\), where
  \(\vec{x} \coloneqq \sum_{i=1, i\neq j}^d c_i \vec{v}_i\) is orthogonal to \(\vec{v}_j\).
  According to Line~\ref{line:find-key-term} of Algorithm~\ref{algo:server} and Condition~\ref{cond:key-term-richness}, we have \(\tilde{\vec{x}}_k^\mathsf{T} \vec{v}_j\geq \beta\) for the selected key term $k$.
  Therefore, we have \((\sum_{i=1}^{d} c_i \vec{v}_i)^\mathsf{T} \vec{v}_j = c_j \geq \beta\), and \(\tilde{\vec{x}}_k \tilde{\vec{x}}_k^\mathsf{T} = (c_j \vec{v}_j+\vec{x})(c_j \vec{v}_j+\vec{x})^\mathsf{T}=c_j^2 \vec{v}_j \vec{v}_j^\mathsf{T} + \vec{x} \vec{x}^\mathsf{T}\).
 By spectral decomposition and  line~\ref{line:check-eigenvalue} in Algorithm~\ref{algo:client}, we have $\vec{M}_m^{p'}=$ $\sum_{i=1}^{d} \lambda_i \vec{v}_i \vec{v}_i^\mathsf{T} + \sum_{j: \lambda_j<h_{p}} \frac{h_{p}-\lambda_j}{C^2} \ab(c_j^2 \vec{v}_j \vec{v}_j^\mathsf{T} + \vec{x} \vec{x}^\mathsf{T})$. Then, $\vec{M}_m^{p'} \succeq $$\sum_{i=1}^{d} \lambda_m \vec{v}_m \vec{v}_m^\mathsf{T} + \sum_{j: \lambda_j<h_{p}} \ab(h_{p} - \lambda_j) \vec{v}_j \vec{v}_j^\mathsf{T}$ $\succeq \sum_{i=1}^{d} \frac{3}{4(1-2^{-2p})d} \vec{v}_m \vec{v}_m^\mathsf{T}.$
The proof concludes by the Loewner order property, stating if \(\vec{A} \succeq \vec{B}\), then \(\lambda_j(\vec{A}) \geq \lambda_j(\vec{B})\).
\end{proof}


\begin{lemma}[Reliability of Estimation Error Bounds]\label{lemma:bad-event}
Define the ``bad'' event \(\mathcal{E}\) where any local agent $m$ at phase \(p\) has:
\[\mathcal{E} = \{\exists m \in \mM, a \in \mathcal{A}_m^{p}, \left| \langle \widehat{\vec{\theta}}_{p} - \vec{\theta}^{*}, \vec{x}_a \rangle \right| > \frac{2^{-p}}{\sqrt{M}}\}.\]
The probability of \(\mathcal{E}\) is bounded by \(\delta\), i.e., \(\Pr[\mathcal{E}] \leq \delta\).

\end{lemma}
\begin{proof}
For any phase \(p\), given \(\vec{G}\)'s definition in Algorithm~\ref{algo:server}, it follows that $\vec{G}=\sum_{s=1}^{p} \sum_{m=1}^{M} \vec{G}_m^p \succeq 2d \log \ab(\frac{2AM\log T}{\delta})  $
$\sum_{m=1}^{M}\underbrace{\ab[\sum_{s=1}^{p} \frac{1}{2^{-2p}}(\sum_{a \in \mathcal{A}_m^{p}} \frac{\vec{x}_a\vec{x}_a^\mathsf{T}}{|\mathcal{A}_m^{p}|} + \sum_{k \in \mathcal{K}_m^{p}} \frac{h_p - \lambda}{\beta^2}\tilde{\vec{x}}_k \tilde{\vec{x}}_k^\mathsf{T})]}_{\triangleq \vec{Q}_m^{p}}.$
  By the Weyl's inequality, we have the lower bound of the smallest eigenvalue of \(\vec{Q}_m^{p}\): $\lambda_{\text{min}}\ab(\vec{Q}_m^{p}) \geq$$ \sum_{s=1}^{p} \frac{1}{2^{-2p}} \lambda_{\text{min}}\ab(\vec{M}_m^{p} + \sum_{k \in \mathcal{K}_m^{p}} \frac{h_p - \lambda}{\beta^2}\tilde{\vec{x}}_k \tilde{\vec{x}}_k^\mathsf{T}) $.  By Lemma~\ref{lemma:lower-bound-of-smallest-eigenvalue}, $\lambda_{\text{min}}\ab(\vec{Q}_m^{p}) \geq$ $ \sum_{s=1}^{p} \frac{1}{2^{-2p}} \frac{3}{4(1-2^{-2p})d}
      \geq \frac{3}{4(1-2^{-2p})d}\sum_{s=1}^{p} \frac{1}{2^{-2p}} = \frac{1}{d\cdot 2^{-2p}}$.
  Based on this, we have $\lambda_{\text{min}} \left( \vec{G} \right) \geq 2^{2p+1}M \log \frac{2AM\log T}{\delta}.$
  According to the concentration of linear regression in Chapter 20.1 of \cite{lattimore-2020-bandit-algorithms} (with the gram matrix refined as   $\vec{G}$ for incorporating information from key terms),  for any \(\delta > 0\), \(s \in [p]\), \(\vec{x} \in \RR^d\), with probability at least \(1-2\delta\), we have
  $\left| \inprod{\widehat{\vec{\theta}}_{s} - \vec{\theta}^{*}}{\vec{x}} \right| \leq \sqrt{2 \|\vec{x}\|_{\vec{G}^{-1}}^2 \log \frac{1}{\delta}}.$ Then, by the Courant-Fischer theorem, 
  with probability at least \(1-\frac{\delta}{AM\log T}\), for any \(m \in \mM\) and all arm \(a \in \mathcal{A}_m^{p}\), we have $\left| \inprod{\widehat{\vec{\theta}}_{p} - \vec{\theta}^{*}}{\vec{x}_a} \right|
    \leq \sqrt{2\|\vec{x}_a\|_{\vec{G}^{-1}}^{2} \log \frac{2AM\log T}{\delta}} $$\leq \sqrt{ \frac{2}{\lambda_{\text{min}}(\vec{G})} \log \frac{2AM\log T}{\delta}} \leq \frac{2^{-p}}{\sqrt{M}}.$
  Finally, by the union bound, $ \Pr\left[\mathcal{E}\right] \leq MPK \frac{\delta}{AM\log T} \leq \delta$ is obtained  with \(P \leq \log T\) (deduced from Section \ref{subsec:analysis}: $ T \geq 2d 2^{2P} \log \frac{AM\log T}{\delta} \geq 2^P$).
\end{proof}

 Now, consider the  ``good'' event \(\mathcal{E}^c\) for agent \(m\) at phase \(p\). Lemma \ref{lemma:bad-event} confirms that the discrepancy for any arm \(a\) in \(\mathcal{A}_m^{p}\):
$
\langle \vec{x}_a - \vec{x}_{a_m^{*}}, \widehat{\vec{\theta}}_{p} \rangle \leq \frac{2^{-p+1}}{\sqrt{M}}.
$
This, combined with line \ref{line:elimination} in Algorithm \ref{algo:client}, supports the following lemma on the arm preservation and performance bound under good event \(\mathcal{E}^c\).

\begin{lemma}[Properties Under Good Event]\label{lemma:optimal-arm-and-performance}
  Under event \(\mathcal{E}^c\),  two key properties are ensured:
  \begin{enumerate}[leftmargin=*]
    \item The locally optimal arm \(\vec{a}_m^{*}\) remains within the active arm set \(\mathcal{A}_m^{p}\), ensuring it is never eliminated.
    \item The performance gap for any arm \(a \in \mathcal{A}_m^{p}\), \(\Delta_{m,a} \triangleq \inprod{\vec{\theta}^{*}}{\vec{x}_{a_m^{*}} - \vec{x}_a}\), is bounded by \(\frac{2^{-p+3}}{\sqrt{M}}\).
  \end{enumerate}
\end{lemma}

Finally, the regret $R_{M}(T)$ 
$ =\sum_{m=1}^{M}\sum_{t=1}^{T} \inprod{\vec{\theta}^{*}}{\vec{x}_{a_m^{*}}-\vec{x}_{a_m,t}}\) can be bounded by $\sum_{m=1}^{M}\sum_{p=1}^{P} \sum_{a \in \mathcal{A}_{m}^{p}} n_{m,a}^p \frac{2^{-p+3}}{\sqrt{M}}$, where $P$ denotes the total number of phases, with probability  $1 -\delta$. Given that $\sum_{a \in \mathcal{A}_{m}^{p}} n_{m,a}^p \leq$$2^{-2p+1}d\log \frac{2AM\log T}{\delta}  + |\mathcal{A}_{m}^{p}|$, we derive that $R_{M}(T) \leq \mathcal{O}\left( d\sqrt{M}\log \frac{AM\log T}{\delta} 2^P \right).$ Furthermore, $ T \geq \sum_{p=1}^{P} \sum_{a \in \mathcal{A}_{m}^{p}} n_{m,a}^p \geq \sum_{p=1}^{P} 2^{-2p+1}d \log \frac{2AM\log T}{\delta}$, resulting in $ T \geq 2d 2^{2P} \log \frac{AM\log T}{\delta}$. As a result, $R_{M}(T) \leq \mathcal{O}\left( \sqrt{dMT\log \frac{KM \log T}{\delta}}\right)$.

\end{proof}

\subsection{Proof of Regret Lower Bound in Theorem \ref{thm:regret}}\label{app:thm-regret}
 \begin{proof}
    Define \(R_{M,\vec{\theta}}^{\pi}(T)\) as the expected cumulative regret of policy \(\pi\) with user preference \(\vec{\theta}\) over \(M\) local agents and time horizon \(T\).  
   Assume that for all local agents $m$, the arms vectors can span $\RR^d$, and $\{\vec{x}_a\}_{a \in \mathcal{A}_m} =$ $ \{\vec{x}_k\}_{k \in \mathcal{K}}=$ $ \set{\vec{e}_1, \vec{e}_2, \dots, \vec{e}_d} \cup \set{(A-d)\text{ arbitrary unit vectors}}$, where \(\vec{e}_i\) is the \(i\)-th standard basis vector in \(\RR^d\).
   Choose \(\vec{\theta}=(\Delta,0,\dots,0)^\mathsf{T}\) (with \(\Delta \in [0,\frac{1}{2}] \) to be determined later). Let random variables \(N_i(t)\), \(\widetilde{N}_j(t)\) be the number of times the \(i\)-th arm and the \(j\)-th key term are selected, by the end of round \(t\).
  Define another user preference \(\vec{\theta}' = (\Delta, 0, \dots, 2\Delta, \dots, 0)^\mathsf{T}\), where $\theta_{\ell} = 2\Delta$ and \(\ell=\argmin_{j>1} \max\ab\{\E_{\vec{\theta}}[N_j(MT)],\E_{\vec{\theta}}[\widetilde{N}_j(MT)]\}\).
    Denote \(N_{m,a}(t)\) as the number of times the \(a\)-th arm is chosen by local agent \(m \in \mM\) after the end of round \(t\).  Given that the optimal arm for \(\vec{\theta}\) is arm 1, pulling other arms increases the expected regret by \(\Delta\). Thus, by Lemma 4.5 in \cite{lattimore-2020-bandit-algorithms}, $ R_{M,\vec{\theta}}^{\pi}(T)
      =$ $\sum_{m=1}^{M} \Delta \sum_{a=2}^{A} \E\nolimits_{\vec{\theta}}[N_{m,a}(T)]]$. Using the inequality \(\E_{\vec{\theta}}[N_j(MT)] \leq \frac{MT}{K-1}\) and \(\E_{\vec{\theta}}[\widetilde{N}_j(MT)] \leq \frac{MT}{K-1}\) and Markov inequality, we get: $R_{M,\vec{\theta}}^{\pi}(T)$ $\geq \Delta \Pr\nolimits_{\vec{\theta}}\ab[MT-\sum_{m=1}^{M} N_{i,1}(T) \geq \frac{MT}{2}] \frac{MT}{2}.$

    For \(\vec{\theta}'\), similarly, we have $R_{M,\vec{\theta}'}^{\pi}(T) \geq$   $ \Delta \Pr\nolimits_{\vec{\theta}'}\ab[\sum_{m=1}^{M} N_{i,1}(T) > \frac{MT}{2}] \frac{MT}{2}.$
    Therefore, applying the Bretagnolle-Huber theorem (Theorem 14.2 in \cite{lattimore-2020-bandit-algorithms}), $R_{M,\vec{\theta}}^{\pi}(T) + R_{M,\vec{\theta}'}^{\pi}(T)$ $\geq \frac{\Delta MT}{4} \exp\ab(-D(\PP_{\theta} \parallel \PP_{\theta'})).$ According to the properties of Kullback–Leibler (KL) divergence, with \(P \sim \mathcal{N}(\mu_1, \sigma^2)\) and \(Q \sim \mathcal{N}(\mu_2, \sigma^2)\), we have $D(\PP_{\vec{\theta}} \parallel \PP_{\vec{\theta}'}) = $ $\E\nolimits_{\vec{\theta}}[N_\ell(MT) + \widetilde{N}_\ell(MT)] D(\mathcal{N}(0,1) \parallel \mathcal{N}(2\Delta,1))$ $=\frac{(\mu_1-\mu_2)^2}{2\sigma^2}$.
    Let \(\Delta=\sqrt{\frac{d-1}{MT}}\), $ \max\ab\{R_{M,\vec{\theta}}^{\pi}(T), R_{M,\vec{\theta}'}^{\pi}(T)\}$ $\geq \frac{R_{M,\vec{\theta}}^{\pi}(T)+R_{M,\vec{\theta}'}^{\pi}(T)}{2} $ $  \geq \frac{e^{-4}}{8} \sqrt{(d-1)MT} = \Omega\ab(\sqrt{dMT}).$ 
  \end{proof}

\subsection{Proof of Theorem~\ref{thm:communication}}
\label{sec:proof-communication}
\begin{proof}
  At each phase \(p\), each local agent \(m\) downloads the following: (a) The key term vector set, containing at most \(d\) feature vectors of dimension \(d\); (b) The repetition counts for each key term \(n_{m,k}^p,\forall k \in \mathcal{K}_m^{p}\), totaling at most $d$ integers; And (3) the estimated preference vector \(\widehat{\vec{\theta}}_{p}\), a \(d\)-dimensional vector. On the other hand,
  the local agent uploads the following: (a) At most \(d\) eigenvalues and their corresponding eigenvectors; (2) The matrix \(\vec{G}_m^{p}\) and \(\vec{W}_m^{p}\), each size of \(d^2\).
Considering that the number of phases is at most \(\log T\), the upload and download costs are both \(\mathcal{O}(d^2M \log T)\).
\end{proof}

\subsection{Proof of Theorem~\ref{thm:conversation} }
\label{sec:proof-conversation}

\begin{proof}
  1) follows directly from line~\ref{line:check-eigenvalue} of Algorithm~\ref{algo:client}.
  For 2), in phase \(p\), the number of arms $n_{m}^p$ pulled by each local agent $m$  is
  $ \sum_{a \in \mathcal{A}_{m}^{p}} n_{m,a}^p=\sum_{a \in \mathcal{A}_{m}^{p}}\left\lceil \frac{2^{2p+1}d}{\mA_m^p}\log \frac{2AM\log T}{\delta} \right\rceil \geq 2^{2p+1}d \log \frac{2AM\log T}{\delta}.$
 And the number of key terms pulled $\widetilde{n}_{m}^p$  by local agent \(m\) is given by:
 $\sum_{k \in \mathcal{K}_m^{p}} n_{m,k}^p$ $=\sum_{j:\lambda_j<h_{p}}\frac{2d\ab(h_{p}-\lambda_j)}{\beta^22^{-2p}} \log\ab(\frac{2AM\log T}{\delta})$ $\leq \sum_{j=1}^d\frac{d\ab(h_{p}-\gamma)}{\beta^22^{-2p-1}} \log\ab(\frac{2AM\log T}{\delta})$.
Thus,  the ratio between the number of key terms and arms for any $m \in \mM$ is upper bounded by
$
    \frac{\widetilde{n}_{m}^p}{n_{m}^p} \leq \frac{h_p-d\gamma}{\beta^2}=\frac{\frac{3}{4(1-2^{-2p})}-d\gamma}{\beta^2} \leq \frac{1-d\gamma}{\beta^2}.
$
\end{proof}

\section{More Experiments}\label{app:experiment}

\begin{table*}[t]
    \centering
    \caption{Cumulative regret for algorithms under two embeddings on the Gretel Multilingual dataset.}\label{tab:realdata2}
        \resizebox{\linewidth}{!}{
    \begin{tabular}{lcccccccc}
        \toprule
        & \multicolumn{4}{c}{Text Embedding Model of Google} & \multicolumn{4}{c}{Text Embedding Model of OpenAI} \\
        \cmidrule(lr){2-5} \cmidrule(lr){6-9}
        Algorithm & $M=4$ & $M=8$ & $M=12$ & $M=16$ & $M=4$ & $M=8$ & $M=12$ & $M=16$ \\
        \midrule
        \small{\texttt{TRIPLE-SH}} 
            & 5285.80$\pm$11.27   & 10695.93$\pm$21.27  & 16045.36$\pm$48.37  & 21341.38$\pm$24.42  
            & 8140.94$\pm$42.59   & 16648.50$\pm$53.13   & 25177.63$\pm$97.66 & 33325.86$\pm$67.87 \\
        \texttt{LinUCB}   
            & 517.23$\pm$4.30     & 1023.67$\pm$16.61    & 1536.14$\pm$41.81    & 2054.01$\pm$52.95    
            & 376.58$\pm$13.79    & 734.99$\pm$31.95     & 1097.72$\pm$36.67    & 1476.26$\pm$31.46  \\
        \texttt{Arm-Con}  
            & 440.27$\pm$4.89     & 887.93$\pm$19.88     & 1312.99$\pm$48.81    & 1762.84$\pm$52.54    
            & 293.60$\pm$14.66    & 582.32$\pm$22.61     & 859.71$\pm$48.61     & 1182.15$\pm$29.75  \\
        \texttt{ConUCB}   
            & 250.83$\pm$3.64     & 496.09$\pm$9.33      & 752.09$\pm$21.03     & 991.33$\pm$26.33     
            & 182.31$\pm$9.19     & 361.50$\pm$17.03     & 549.97$\pm$26.20     & 732.82$\pm$16.66   \\
        \tiny{\texttt{ConLinUCB-BS}} 
            & 1047.63$\pm$13.69   & 2059.29$\pm$34.84    & 3096.94$\pm$67.77    & 4135.90$\pm$90.78    
            & 768.99$\pm$22.76    & 1474.07$\pm$63.89    & 2196.14$\pm$96.14    & 2949.82$\pm$59.38  \\
        \tiny{\texttt{ConLinUCB-MCR}}
            & 387.31$\pm$4.69     & 777.48$\pm$15.07     & 1163.03$\pm$26.54    & 1558.77$\pm$49.91    
            & 328.49$\pm$20.06    & 646.82$\pm$26.75     & 967.80$\pm$33.11     & 1291.16$\pm$30.39  \\
        \tiny{\texttt{ConLinUCB-UCB}}
            & 1194.96$\pm$8.88    & 2386.41$\pm$31.31    & 3582.20$\pm$63.30    & 4782.29$\pm$97.81    
            & 917.31$\pm$31.89    & 1808.25$\pm$58.41    & 2699.60$\pm$71.09    & 3612.74$\pm$55.74  \\
      \texttt{MACO} (Ours)    
            & \textbf{36.92$\pm$0.70} & \textbf{73.45$\pm$1.39} & \textbf{110.07$\pm$3.47} & \textbf{147.15$\pm$5.42} 
            & \textbf{56.13$\pm$2.37} & \textbf{114.70$\pm$4.54} & \textbf{173.15$\pm$4.75} & \textbf{229.30$\pm$4.63} \\
        \bottomrule
    \end{tabular}
    }

\end{table*}

\textbf{Response Settings.} We explore the implementation of two response settings using the aforementioned embedding models based on a real-world dataset and an open-source LLM:

 \ding{182}  Following the style classification by \cite{köpf2023openassistantconversationsdemocratizing}, we gather a comprehensive set of 13 keywords representing diverse styles such as ``humorous" and ``helpful", each representing a key term. These keyword styles generate 510 unique combinations, each forming an  ``arm'', where each arm represents a potential style of LLM response. Users have varying priorities for different keyword combinations, and their preference vector $\vec{\theta}$ has the highest cosine similarity with the feature vector $\vec{x}$ of their most favored keyword style (which is \textit{unknown} to the algorithms in advance).
To generate these feature vectors $\vec{x}$ for LLM responses and user preference vectors $\vec{\theta}$ on keywords, we utilize two previously mentioned embedding models. We select the top $d = 256$ dimensions as the feature representation and normalize them into a more concise and efficient dimensional space. The reward is obtained from the cosine similarity between a specific user's preference vector and the feature vector of the selected arm, and the optimal LLM response is defined as the one with the largest reward according to \cite{reimers-gurevych-2019-sentence}.

 \ding{183} Prompt engineering is utilized to construct the initiatory set of responses offline. Following \cite{sahoo2024systematicsurveypromptengineering}, we select a set of keyword styles (i.e., key terms) rich in personal identifiers to establish a diverse style collection, including terms like \textit{helpful and creative use of emojis}. Two keyword styles are jointly selected for each query, which forms a style-specific question to the LLM, ensuring focused and relevant responses. We utilize Llama-3-8B-Instruct \cite{Ollama} to generate corresponding responses.
Each prompt triggers a specific response from the LLM, with each user preference dictating a response styled according to their selected input. For example,
User: "Tell me a joke."  
The response Arm: A variety of jokes under different styles.  
Key terms: Different styles.  
We construct a total arm set of $|\mathcal{A}| = 455$ responses by formulating responses to five different questions, each with two keyword styles. This extensive collection allows for a comprehensive mapping of responses to specific user preferences, effectively forming a set of $455$ user-preference pairs.
Regarding the reward definition, the feature vector extraction, and subsequent steps, we apply the same procedures described above.

\textbf{Cumulative Regret on Gretel Multilingual Dataset.}
We assess the performance of \textsc{cadi} by measuring its cumulative regret against seven baseline algorithms using the Gretel Multilingual dataset \cite{gretel-synthetic-multilingual-llm-prompts-2024} over \(T=1000\) time steps. The evaluation employs Google and OpenAI text embedding models across varying numbers of local agents (\(M=4, 8, 12, 16\)). Consistent with findings from the StyleEval dataset (\Cref{tab:realdata1}), where \textsc{cadi} achieved the lowest regret, surpassing the best baseline (\texttt{ConUCB}) by at least 79.56\% for \(M=4\) with Google embeddings and 71.08\% with OpenAI embeddings, similar trends emerge on the Gretel Multilingual dataset. As shown in  \Cref{tab:realdata2}, \textsc{cadi} reduces regret by at least 80.07\% (Google embeddings) and 69.23\% (OpenAI embeddings) compared to \texttt{ConUCB} for \(M=4\). Across both datasets, \textsc{cadi} demonstrates sublinear regret scaling with respect to \(M\), confirming the robustness of its multi-agent framework for online LLM response evaluation and selection in diverse, multilingual contexts.

\textbf{Learning Efficiency across Multi Agents}.  \(\text{\maco}\) capitalizes on the aggregated data from all local agents, managing to scale its regret according to \(\mathcal{\widetilde{O}}(\sqrt{dMT})\). This scaling significantly dampens the increase in regret, demonstrating the effectiveness of our algorithm's multi-agent approach for online LLM response identification. A clearer depiction of this regret trend is shown in Fig. \ref{fig:result_diff_clients}, where \texttt{TRIPLE-SH} is excluded due to its inferior performance under Scenario Setting 1 with the Google's model and \(T=100000\).

\begin{table}[!ht]
    \centering
    \caption{Execution time (s) ($\pm$ standard deviation) on four settings.}
    \label{tab:time}
\resizebox{0.99\columnwidth}{!}{
    \begin{tabular}{|l|lll|}
    \hline
         \makebox[0.12 \textwidth][c]{\diagbox{\tiny{Setting}}{\tiny{Algorithm}}} & \underline{MACO (w/o G)} & MACO (w/G) & ConLinUCB-BS \\ \hline
       Setting (a)  &$ 2.576\pm0.047$ & $9.766\pm2.709$ & $18.124\pm0.111$ \\ 
        Setting (b) & $2.546\pm0.039$ & $14.272\pm7.107$ & $18.056\pm0.065$ \\ 
        Setting (c) & $2.576\pm0.085$ & $6.369\pm2.832$ &$ 17.926\pm0.095$ \\ 
        Setting (d) &$ 2.661\pm0.056$ & $6.270\pm2.013$ &$ 17.919\pm0.072$ \\ \hline
    \end{tabular}}
\end{table}

\textbf{Execution Time vs. Reward with Deviation.} 
We now present a full assessment of the execution time with statistical information of our algorithm, termed \texttt{MACO w/o G} for emphasis, against \(\texttt{ConLinUCB-BS}\) under conditions of \(T = 5000\) across 6 phases ($A=40,M=4$), and compare it with \(\texttt{MACO w/G}\), which continues to employ the traditional \(\textit{G-optimal design}\). 
The results, detailed in \Cref{tab:time}, show that our algorithm significantly reduces execution time by avoiding the \(\textit{G-optimal design}\) and leveraging data aggregation from multiple local agents to accelerate the learning process. \Cref{tab:time}  further illustrates that \texttt{MACO w/o G} exhibits the lowest deviation since the \textit{information matrix} \(\vec{M}_m^p\) is no longer dependent on a continuously adjusted distribution policy (see Eq. (\ref{eq:G-optimal})).
Additionally, the results in \Cref{tab:reward} show that the average reward for \(\text{\texttt{MACO w/o G}}\) matches that of \(\text{\texttt{MACO w/G}}\), demonstrating that our conversational approach maintains performance while replacing the traditional \(\textit{G-optimal design}\) with a more practical, conversation-based design. 

\begin{table}[!t]
    \centering
    \caption{Average reward ($\pm$ standard deviation) on four settings.}
    \label{tab:reward}
\resizebox{0.99\columnwidth}{!}{
    \begin{tabular}{|l|lll|}
    \hline
         \makebox[0.12 \textwidth][c]{\diagbox{\tiny{Setting}}{\tiny{Algorithm}}} & \underline{MACO (w/o G)} & MACO (w/G) & ConLinUCB-BS \\ \hline
       Setting (a)  &$ 61.849\pm0.558$ & $61.847\pm0.565$ & $59.811\pm0.610$ \\ 
        Setting (b) & $61.605\pm0.642$ & $61.591\pm0.649$ & $59.663\pm0.671$ \\ 
        Setting (c) & $47.405\pm0.977$ & $47.381\pm1.002$ &$ 46.104\pm0.962$ \\ 
        Setting (d) &$41.770\pm0.349$ & $41.858\pm0.412$ &$ 40.720\pm0.349$ \\ \hline
    \end{tabular}}
\end{table}

\textbf{Ablation Study.} \Cref{tab:parameter} reveals that the introduction of the coverage parameter \(\beta\)  in our design has a minimal impact on the outcomes, contrasting with the significant influence exerted by the statistical confidence parameter \(\delta\). 
Importantly, 
 \(\delta\) is an inherent component of the statistical theory and is not an additional hyperparameter introduced in our model,
which is established by the convention \cite{lattimore-2020-bandit-algorithms}. This observation underscores that our framework does not introduce new dependencies on parameters under the bandit algorithm framework.
 
\begin{table}[!ht]
    \centering
    \caption{Cumulative regret under $T=100000,A=40,M=4$.}
    \label{tab:parameter}
    
\resizebox{0.99\columnwidth}{!}{
    \begin{tabular}{|l|llll|}
    \hline
        \makebox[0.12\textwidth][c]{\diagbox{\tiny{Parameter}}{\tiny{Setting}}} & Setting (a)& Setting (b) & Setting (c) & Setting (d) \\\hline
       $\beta=1.0,\mathbf{\delta=0.1}$  &$ 20213.773$ & $16277.413$ & $15033.483$ & $8261.335$\\ 
        $\beta=0.9,\mathbf{\delta=0.05}$ & $21439.795$ & $17205.540$ & $16039.654$ &$ 8772.119$\\ 
           $\beta=0.8,\mathbf{\delta=0.05}$ & $21430.625$ & $17215.402$ & $16033.950$ & $8770.108$\\
        $\beta=0.9,\mathbf{\delta=0.15}$ & $19495.106$ & $15734.833$ & $15092.586$ & $7962.415$\\
              $\beta=0.8,\mathbf{\delta=0.15}$ & $19492.169$ & $15738.395$ & $15094.809$ &$ 7961.321$\\ 
     
        \hline
    \end{tabular}}
     
\end{table}

\end{document}